\documentclass[10pt,conference]{IEEEtran}
\usepackage{amsmath,amssymb}
\usepackage{graphicx}
\usepackage{mathrsfs}
\usepackage{tikz}
\usepackage{booktabs}
\usepackage{floatflt}
\usepackage{enumerate}
\usepackage{psfrag}	
\usepackage{array}
\usepackage{multirow,hhline}
\usepackage{exscale}
\usepackage{color}
\usepackage{colortbl}
\usepackage{epsfig}
\usepackage{cite}
\usepackage{amsthm}

\renewcommand{\vec}[1]{\ensuremath{\mathbf{#1}}}

\newtheorem{theorem}{Theorem}

\begin{document}

\IEEEoverridecommandlockouts
\title{Sub-optimality of Treating Interference as Noise in the Cellular Uplink}
\author{
\IEEEauthorblockN{Anas Chaaban and Aydin Sezgin}
\IEEEauthorblockA{Emmy-Noether Research Group on Wireless Networks\\
Institute of Telecommunications and Applied Information Theory\\
Ulm University, 89081 Ulm, Germany\\
Email: {anas.chaaban@uni-ulm.de, aydin.sezgin@uni-ulm.de}}
\thanks{%
This work is supported by the German Research Foundation, Deutsche
Forschungsgemeinschaft (DFG), Germany, under grant SE 1697/3.%
}
}

\maketitle


\begin{abstract}
Despite the simplicity of the scheme of treating interference as noise (TIN), it was shown to be sum-capacity optimal in the Gaussian 2-user interference channel in \cite{ShangKramerChen,MotahariKhandani,AnnapureddyVeeravalli}. In this paper, an interference network consisting of a point-to-point channel interfering with a multiple access channel (MAC) is considered, with focus on the weak interference scenario. Naive TIN in this network is performed by using Gaussian codes at the transmitters, joint decoding at the MAC receiver while treating interference as noise, and single user decoding at the point-to-point receiver while treating both interferers as noise. It is shown that this naive TIN scheme is never optimal in this scenario. In fact, a scheme that combines both time division multiple access and TIN outperforms the naive TIN scheme. An upper bound on the sum-capacity of the given network is also derived.
\end{abstract}

\section{Introduction}
A simple communication strategy in wireless interference networks is treating interference as noise (TIN). In this scheme, the receivers' strategy is the same as if there were no interference, i.e., interference is ignored. Although this scheme seems very trivial and too simple to be optimal, its optimality in terms of sum-capacity in the 2-user Gaussian interference channel (IC) was recently shown in \cite{ShangKramerChen,MotahariKhandani,AnnapureddyVeeravalli} in the so-called noisy interference regime, where transmitters use Gaussian codes to encode their messages. If receivers in IC's with more than two users want to apply TIN, then each receiver has to treat more than one interferer as noise, contrary to the 2-user IC where each receiver faces only one interferer. Can it be optimal to treat more than one interferer as noise? This was answered in \cite{ShangKramerChen_KUserIC} where the optimality of TIN in a noisy interference regime in the K-user IC with $K\geq2$ was also established.

In this paper, we consider a network consisting of a point-to-point (P2P) channel, and a multiple access channel (MAC) interfering with each other. We call this network the partial IMAC (PIMAC), to distinguish it from the IMAC studied in \cite{SuhTse} and \cite{ChaabanSezgin_IMAC}. Such a setup might arise if a P2P communication system uses the same communication medium as a cellular uplink. The PIMAC setup was studied recently in \cite{ChaabanSezgin_EW2011} where its capacity in strong and very strong interference cases was obtained and a sum capacity upper bound was derived, and in \cite{ZhuShangChenPoor} where an achievable rate region for the discrete memoryless Z-PIMAC (partially connected PIMAC) was provided, which achieves the capacity of the Z-PIMAC with strong interference. A cognitive variant of MAC interference networks was also studied in \cite{DevroyeMitranTarokh_ISIT}.

The PIMAC combines properties from both the 2-user IC and the K-user IC, namely, one receiver receives two interference signals and the other receiver receives one interference signal. The capacity achieving scheme in the interference free MAC is known, and so is that in the interference free P2P channel \cite{CoverThomas}. Now, in a Gaussian PIMAC, if the receivers use their interference free strategies and ignore interference, then we call this naive TIN scheme the (SD-TIN) referring to the capacity achieving successive decoding (SD) scheme in the interference free MAC. In this case, receivers proceed with decoding using their interference-free-optimal decoders while treating interference as noise. Interestingly, we show in this paper that this scheme is never optimal, contrary to the 2-user IC and the K-user IC. To show this, we examine an alternative scheme which combines time division multiple access (TDMA) and TIN, which we call TDMA-TIN. We show that this scheme achieves a sum-rate which is larger that that of the SD-TIN scheme for all possible channel parameters. Therefore, TDMA-TIN outperforms SD-TIN, and the SD-TIN scheme is never optimal. Interestingly, while in the interference free MAC, the capacity achieving scheme of SD performs the same as TDMA in terms of sum-capacity, the same is not true in the presence of interference with TIN. We also introduce another scheme where transmitters can transmit at arbitrary power lower than their power constraint, and where the MAC receiver uses SD with TIN, and the P2P receiver uses TIN. We call this scheme the PC-TIN scheme. We introduce a sum-capacity upper bound and then compare the achievable sum-rate of the given schemes with this upper bound and the upper bound obtained in \cite{ChaabanSezgin_EW2011}. 

To this end, the organization of the paper is as follows. In Section \ref{Model} we describe the system model of the PIMAC. The three schemes of SD-TIN, TDMA-TIN, and PC-TIN are described in Section \ref{LowerBound}. The TDMA-TIN scheme and the SD-TIN scheme are compared with each other in section \ref{Comparison}. We provide a new sum-capacity upper bound in Section \ref{UpperBounds}. We discuss the results in Section \ref{Discussion} and we conclude in Section \ref{Conclusion} with an outlook.

\section{System Model}
\label{Model}

The PIMAC is a network consisting of a multiple access channel (MAC) and a point-to-point (P2P) channel interfering with each other as shown in Figure \ref{pIMAC}. The received signals at the two receivers of  the PIMAC can be written as
\begin{align}
Y_1&=X_1+X_2+h_{31}X_3+Z_1,\\
Y_2&=h_{12}X_1+h_{22}X_2+X_3+Z_2,
\end{align}
where the channel coefficients $h_{ij}$ from transmitter $i\in\{1,2,3\}$ to receiver $j\in\{1,2\}$ is real valued, $X_i$ is a random variable representing the transmit signal of transmitter $i$, and $Z_j$ is an independent and identically distributed (i.i.d.) additive white Gaussian noise (AWGN) with $Z_j\sim\mathcal{N}(0,1)$. It must be noted that the normalization of the noise variance and the direct channel gains is done without loss of generality.

\begin{figure}[t]
\centering
\includegraphics[width=0.8\columnwidth]{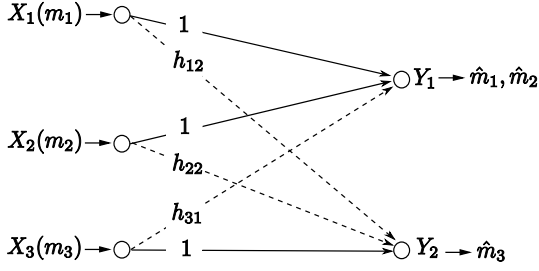}
\caption{A multiple access channel interfering with a point-to-point channel. The first receiver can be a base station in a cellular uplink system with two users, whereas the point-to-point channel can resemble a secondary system using the same communications medium.}
\label{pIMAC}
\end{figure}
 
Each transmitter uses an encoding function $f_i$ to encode its message $m_i$ chosen independently from the set $\mathcal{M}_i=\{1,\dots,2^{nR_i}\}$ into a codeword  of length $n$ symbols $X_i^n\in\mathbb{R}^n$ $$X_i^n=(X_{i,1},\dots,X_{i,n})$$ such that $\frac{1}{n}\sum_{j=1}^n\mathbb{E}[X_{i,j}^2]\leq P_i$. 

After receiving $Y_1^n$, the first receiver decodes $m_1$ and $m_2$ by using a decoding function $g_1$, i.e. $(\hat{m}_1,\hat{m}_2)=g_1(Y_1^n)$. The second receiver decodes $m_3$ from $Y_2^n$ by using a decoding function $g_2$, i.e. $\hat{m}_3=g_2(Y_2^n)$. The decoding results in an error $E_i$ if $m_i\neq\hat{m}_i$, $i\in\{1,2,3\}$. 

The encoding functions, decoding functions, and message sets define an $(n,2^{nR_1},2^{nR_2},2^{nR_3})$ code for the PIMAC, and induces an average probability of decoding error given by 
\begin{align}
P^{(n)}=\frac{1}{2^{nR_\Sigma}}\sum_{\vec{m}\in\mathcal{M}_1\times\mathcal{M}_2\times\mathcal{M}_3}P\left(\bigcup_{i=1}^3E_i\right),
\end{align}
where $R_\Sigma=\sum_{i=1}^3R_i$ and $\vec{m}=(m_1,m_2,m_3)$. We say that a rate triple $(R_1,R_2,R_3)$ is achievable if there exists a sequence of $(n,2^{nR_1},2^{nR_2},2^{nR_3})$ codes such that $P^{(n)}\to0$ as $n\to\infty$. The capacity region of the PIMAC denoted $\mathcal{C}$ is the set of all achievable rate triples, and the sum capacity is the highest achievable sum-rate 
\begin{align}
C_\Sigma=\max_{(R_1,R_2,R_3)\in\mathcal{C}}R_\Sigma.
\end{align}

\section{Treating interference as noise}
\label{LowerBound} 
Now, we introduce achievable schemes for the PIMAC that simply ignore the interference by treating it as additional noise. Namely, we introduce the successive decoding with TIN (SD-TIN) scheme, then the time division multiple access with TIN (TDMA-TIN) scheme, and finally, we introduce a variant of the SD-TIN scheme with power control (PC-TIN).

\begin{figure}[t]
\centering
\psfragscanon
\psfrag{x}[l]{$R_1$}
\psfrag{y}[b]{$R_2$}
\psfrag{L1}[l]{\footnotesize{MAC capacity region}}
\psfrag{L2}[l]{\footnotesize{TDMA achievable region}}
\includegraphics[width=0.7\columnwidth]{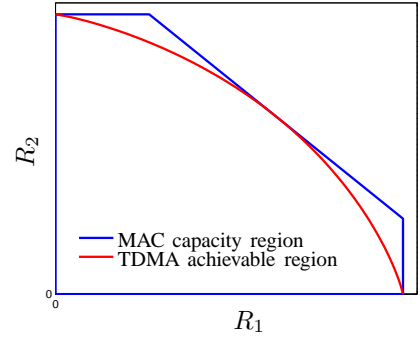}
\caption{The MAC capacity region achieved using SD, and the achievable region using TDMA. Both the SD scheme and the TDMA scheme achieve the sum-capacity of the MAC.}
\label{MAC_Reg}
\end{figure}

\subsection{SD-TIN}
It is known that using i.i.d. Gaussian codes and successive decoding (SD) achieves the capacity region of the Gaussian MAC \cite{CoverThomas} (cf. Figure \ref{MAC_Reg}). In the SD-TIN scheme, we assume that the MAC receiver in the PIMAC uses the same decoding strategy as if there were no interference, i.e. it uses SD and treats interference as noise. That is, the MAC receiver decodes one desired signal, say $X_1$, while treating $X_2$ and $X_3$ as noise, then it subtracts the contribution of $X_1$ from its received signal $Y_1$ and decodes $X_2$ while treating $X_3$ as noise. The second receiver decodes $X_3$ while treating both $X_1$ and $X_2$ as noise. Using this scheme, the following rates are achievable

\begin{align}
\label{IRB1}
R_1&\leq\frac{1}{2}\log\left(1+\frac{P_1}{1+h_{31}^2P_3}\right)\\
\label{IRB2}
R_2&\leq\frac{1}{2}\log\left(1+\frac{P_2}{1+h_{31}^2P_3}\right)\\
\label{SRB}
R_1+R_2&\leq\frac{1}{2}\log\left(1+\frac{P_1+P_2}{1+h_{31}^2P_3}\right)\\
\label{IRB3}
R_3&\leq\frac{1}{2}\log\left(1+\frac{P_3}{1+h_{12}^2P_1+h_{22}^2P_2}\right)
\end{align}

Inequalities \eqref{IRB1}-\eqref{SRB} represent the rate constraints for the MAC users, while \eqref{IRB3} represents the rate constraint for the P2P user. The effective noise power in \eqref{IRB1}-\eqref{IRB3} is $1+h_{31}^2P_3$ and $1+h_{12}^2P_1+h_{22}^2P_2$ which is the sum of the noise power and the interference power at the receivers. Since the sum-rate constraint \eqref{SRB} is always smaller than the sum of \eqref{IRB1} and \eqref{IRB2}, the achievable sum rate is the sum of \eqref{SRB} and \eqref{IRB3}. Thus, the achievable sum-rate using the SD-TIN scheme is given by
\begin{align}
\label{SD-TDMA_SR}
R_\Sigma^{SD-TIN}&=\frac{1}{2}\log\left(1+\frac{P_1+P_2}{1+h_{31}^2P_3}\right)\nonumber\\
&\quad+\frac{1}{2}\log\left(1+\frac{P_3}{1+h_{12}^2P_1+h_{22}^2P_2}\right)
\end{align}

\subsection{TDMA-TIN}
In this scheme, the users of the MAC component of the PIMAC use TDMA. Each user of this MAC uses the channel for a fraction of time, say $\alpha$ for user 1 and $1-\alpha$ for user 2, with $\alpha\in[0,1]$. This schemes transforms the PIMAC into two 2-user Gaussian IC's, each of which is active for a fraction of the time. User 1 uses i.i.d. $X_1\sim\mathcal{N}(0,P_1/\alpha)$ and user 2 uses i.i.d.  $X_2\sim\mathcal{N}(0,P_2/(1-\alpha))$. In the first fraction $\alpha$ of the time, transmitters 1 and 3 are active, and each of the receivers decodes its desired signal while treating interference as noise. Similarly, in the remaining fraction $1-\alpha$ of the time, transmitters 2 and 3 are active and the receivers also treat interference as noise. As a result, the following sum-rate is achievable using the TDMA-TIN scheme
\begin{align}
R_\Sigma^{TDMA-TIN}=\max_{\alpha\in[0,1]}A(\alpha)+B(\alpha)
\end{align}
where
\begin{align}
A(\alpha)&=\frac{\alpha}{2}\log\left(1+\frac{\frac{P_1}{\alpha}}{1+h_{31}^2P_3}\right)\nonumber\\
&\quad+\frac{1-\alpha}{2}\log\left(1+\frac{\frac{P_2}{1-\alpha}}{1+h_{31}^2P_3}\right)\nonumber\\
B(\alpha)&=\frac{\alpha}{2}\log\left(1+\frac{P_3}{1+\frac{h_{12}^2P_1}{\alpha}}\right)\nonumber\\
&\quad+\frac{1-\alpha}{2}\log\left(1+\frac{P_3}{1+\frac{h_{22}^2P_2}{1-\alpha}}\right).
\end{align}

Here, $A(\alpha)$ is the sum-rate achieved by the MAC users, and $B(\alpha)$ is the rate achieved by the P2P user.

\subsection{PC-TIN}
The scheme used here is similar to the SD-TIN, with optimization with respect to the transmit power of each transmitter. In this case, we allow the transmitters to use any transmit power $p_i\leq P_i$, $i\in\{1,2,3\}$. Although the transmitters should transmit at full power at weak interference, this is not the case as interference increases. In some cases, decreasing the transmit power of one of the transmitters can increase the achievable sum rate \cite{CharafeddineSezginPaulraj}. 

Again, transmitters use codes with i.i.d. Gaussian $\mathcal{N}(0,p_i)$ symbols where $p_i\leq P_i$. The receivers treat interference as noise. The achievable sum rate is given by
\begin{align}
\label{PC-TDMA_SR}
R_\Sigma^{PC-TIN}&=\max_{p_i\leq P_i} C(p_1,p_2,p_3)
\end{align}
where
\begin{align}
C(p_1,p_2,p_3)&=\frac{1}{2}\log\left(1+\frac{p_1+p_2}{1+h_{31}^2p_3}\right)\nonumber\\
\label{C_PC-TIN}
&\quad+\frac{1}{2}\log\left(1+\frac{p_3}{1+h_{12}^2p_1+h_{22}^2p_2}\right)
\end{align}
Notice that the SD-TIN scheme is a special case of the PC-TIN scheme. We separate these two cases to distinguish the naive TIN scheme (with full power) from the more general PC-TIN.

\section{TDMA-TIN outperforms SD-TIN}
\label{Comparison}
Let us now compare the SD-TIN scheme and the TDMA-TIN scheme with each other. First, we notice that the term $A(\alpha)$ is the achievable sum-rate using TDMA with time sharing parameter $\alpha$ in a MAC with a noise variance of $1+h_{31}^2P_3$. This sum-rate is maximized if we choose 
\begin{align}
\label{AlphaStar}
\alpha=\frac{P_1}{P_1+P_2}\triangleq\alpha^*,
\end{align}
and its maximum value is
\begin{align}
A(\alpha^*)=\frac{1}{2}\log\left(1+\frac{P_1+P_2}{1+h_{31}^2P_3}\right)
\end{align}
which is equal to the first term of $R_\Sigma^{SD-TIN}$ in \eqref{SD-TDMA_SR}. On the other hand, the function $B(\alpha)$ can be shown to satisfy the following
\begin{align}
\left.\frac{d B(\alpha)}{d\alpha}\right|_{\alpha=\alpha'}&=0\\
\frac{d^2 B(\alpha)}{d\alpha^2}&\geq0\quad\forall \alpha\in[0,1]
\end{align}
where
\begin{align}
\alpha'=\frac{h_{12}^2P_1}{h_{12}^2P_1+h_{22}^2P_2}.
\end{align}

Thus, $B(\alpha)$ is convex in $\alpha$ and achieves its minimum at $\alpha'$. This minimum is equal to 
\begin{align}
\min_{\alpha\in[0,1]}B(\alpha)&=B(\alpha')\nonumber\\
&=\frac{1}{2}\log\left(1+\frac{P_3}{1+h_{12}^2P_1+h_{22}^2P_2}\right)
\end{align}

which is the second term of $R_\Sigma^{SD-TIN}$ in \eqref{SD-TDMA_SR}. Therefore, 
\begin{align}
R_\Sigma^{TDMA-TIN}&=\max_{\alpha\in[0,1]}A(\alpha)+B(\alpha)\\
&\stackrel{(a)}{\geq} A(\alpha^*)+B(\alpha^*)\\
&\stackrel{(b)}{\geq} A(\alpha^*)+B(\alpha')\\
&= R_\Sigma^{SD-TIN}.
\end{align}
where $(a)$ follows by replacing the optimal value of $\alpha$ by $\alpha^*$ defined in \eqref{AlphaStar}, and $(b)$ follows since $B(\alpha)$ is minimized by $\alpha'$. Equality in $(b)$ holds if $\alpha^*=\alpha'$, i.e., $h_{12}^2=h_{22}^2$. Otherwise, the inequality in $(b)$ can be replaced with a strict inequality. As a result, by choosing $\alpha=\alpha^*$, the TDMA-TIN scheme outperforms the SD-TIN scheme. Notice that higher rates can be achieved by optimizing with respect to $\alpha$.

\section{Sum-capacity upper bounds}
\label{UpperBounds}
In the following, two upper bounds are introduced serving as benchmarks for the schemes we have investigated in the previous section. In \cite{ChaabanSezgin_EW2011}, a sum-capacity upper bound for the PIMAC was given as follows
\begin{align}
\label{UB1}
C_\Sigma&\leq\overline{C}_{\Sigma 1}\\
&\triangleq\min_{\substack{|\rho_1|,|\rho_2|\leq1,\\
\eta_1^2\leq1-\rho_2^2,\\
\eta_2^2\leq1-\rho_1^2}} I(X_{1G},X_{2G};Y_{1G},S_{1G})+I(X_{3G};Y_{2G},S_{2G})\nonumber
\end{align}
where $X_{iG}\sim\mathcal{N}(0,P_i)$ for $i\in\{1,2,3\}$, $Y_{jG}$ for $j\in\{1,2\}$ is the channel output when the input is $X_{iG}$\footnote{The bound in \cite{ChaabanSezgin_EW2011} derived for a PIMAC with a constraint of the form $\mathbb{E}[X_{i,j}^2]\leq P_i$ can be shown to hold in our case with $\frac{1}{n}\sum_{j=1}^n\mathbb{E}[X_{i,j}^2]\leq P_i$ by using convexity arguments.}. The upper bound is obtained by giving the receivers the following genie signals
\begin{align}
S_{1G}&=h_{12}X_{1G}+h_{22}X_{2G}+\eta_1W_1\\
S_{2G}&=h_{31}X_{3G}+\eta_2W_2,
\end{align}
as side information to facilitate the derivation of the upper bound, with noise $W_j\sim\mathcal{N}(0,1)$, noise scaling $\eta_j\in\mathbb{R}$, and where the genie noise is correlated with the noise at the receivers as $\mathbb{E}[W_jZ_j]=\rho_j$.

This upper bound is nearly tight if the interference is very weak. Here we give another sum-capacity upper bound based on a Z channel approach.
\begin{theorem}
The sum-capacity of the PIMAC with $h_{31}^2\leq1$ is upper bounded by
\begin{align}
\label{UB2}
C_{\Sigma}\leq\overline{C}_{\Sigma 2}\triangleq\frac{1}{2}\log\left(1+\frac{P_1+P_2}{1+h_{31}^2P_3}\right)+\frac{1}{2}\log(1+P_3).
\end{align}
\end{theorem}
\begin{proof}
We begin with Fano's inequality \cite{CoverThomas} to bound the achievable rates by
\begin{align}
n(R_1+R_2)&\leq I(X_1^n,X_2^n;Y_1^n)+n\varepsilon_{1n}\\
nR_3&\leq I(X_3^n;Y_2^n)+n\varepsilon_{2n}
\end{align}
where $\varepsilon_{1n}$ and $\varepsilon_{2n}$ approach zero as $n$ increases. We proceed as follows
\begin{align}
nR_3&\leq I(X_3^n;Y_2^n)+n\varepsilon_{2n}\\
&\leq I(X_3^n;Y_2^n,X_1^n,X_2^n)+n\varepsilon_{2n}\\
&\leq I(X_3^n;Y_2^n|X_1^n,X_2^n)+n\varepsilon_{2n}
\end{align}
which follows by giving $X_1^n$ and $X_2^n$ as side information to the second receiver, and then using the chain rule of mutual information and the independence of $X_1^n$, $X_2^n$, and $X_3^n$ since each transmitter chooses its message independently. Then
\begin{align}
n(R_\Sigma-\varepsilon_n)&\leq I(X_1^n,X_2^n;Y_1^n)+I(X_3^n;Y_2^n|X_1^n,X_2^n),
\end{align}
where $\varepsilon_n=\varepsilon_{1n}+\varepsilon_{2n}\to0$ as $n\to\infty$. Using the definition of mutual information, we get
\begin{align}
n(R_\Sigma-\varepsilon_n)&\leq h(Y_1^n)-h(Y_1^n|X_1^n,X_2^n)\nonumber\\
&\quad+h(Y_2^n|X_1^n,X_2^n)-h(Y_2^n|X_1^n,X_2^n,X_3^n)\nonumber\\
&= h(Y_1^n)-h(h_{31}X_3^n+Z_1^n)\nonumber\\
\label{InEq1}
&\quad+h(X_3^n+Z_2^n)-h(Z_2^n).
\end{align}
Consider the first term in \eqref{InEq1}. This can be bounded as follows
\begin{align}
h(Y_1^n)&\stackrel{(a)}{=} \sum_{i=1}^nh(Y_{1i}|Y_{1}^{i-1})\\
&\stackrel{(b)}{\leq} \sum_{i=1}^nh(Y_{1i})\\
\label{InEq2}
&\stackrel{(c)}{\leq} nh(Y_{1G}),
\end{align}
where $(a)$ follows by using the chain rule of entropy, $(b)$ follows since conditioning does not increase entropy, $(c)$ follows the Gaussian distribution $X_{iG}\sim\mathcal{N}(0,P_i)$ maximizes the differential entropy term under a covariance constraint \cite{CoverThomas}. The second and third term in \eqref{InEq1} can be bounded by
\begin{align}
\label{InEq3}
-h(h_{31}X_3^n+Z_1^n)&+h(X_3^n+Z_2^n)\\
&\stackrel{(d)}{\leq} -nh(h_{31}X_{3G}+Z_1)+nh(X_{3G}+Z_2),\nonumber
\end{align}
using the worst case noise lemma \cite{DiggaviCover} which applies since $h_{31}^2\leq1$ (cf. \cite{EtkinTseWang} e.g.). Finally
\begin{align}
\label{InEq4}
h(Z_2^n)=nh(Z_2),
\end{align}
since the noise is i.i.d.. Combining \eqref{InEq2}-\eqref{InEq4} and taking the limit as $n\to\infty$, we obtain the desired upper bound.
\end{proof}

\section{Discussion}
\label{Discussion}
In order to highlight the difference with the 2-user IC, we start by recalling some facts. In the 2-user Gaussian IC, each of the constituent P2P channels can use the same  communication strategy as if there were no interference, i.e., single user detection with TIN. In this case, the transmitters use full power for their transmission. This scheme was shown to be sum-capacity optimal in \cite{ShangKramerChen,MotahariKhandani,AnnapureddyVeeravalli} in the so-called noisy interference regime.

Consider the following two sum-capacity achieving strategies in the MAC: SD and TDMA. In SD, the MAC users transmit at full power and the receiver decodes the signals successively one after another. In TDMA, the users share the time, and each user uses all the available power in the time slot assigned to him. Both schemes perform the same in the MAC in terms of sum-rate and achieve its sum-capacity.

In the presence of interference, as in the PIMAC, the nodes of the MAC can still use the same schemes, and treat the interference at the MAC receiver as noise. So does the P2P channel receiver, it can also treat interference as noise. However, in this case, TDMA-TIN achieves a higher sum-rate compared to SD-TIN. Thus, SD-TIN can not be sum-capacity optimal. 

Contrary to the 2-user Gaussian IC where transmitting at full power and treating interference as noise can be optimal in the noisy interference regime, this is not the case in the PIMAC, even if we have very weak interference. This is illustrated in Figure \ref{SumRate1} where the achievable sum-rate with SD-TIN and TDMA-TIN is plotted  as a function of $h=h_{12}=h_{31}$ for a PIMAC with $P_1=P_2=P_3=10$ and $h_{22}=0.2$. In the same figure, we plot the achievable sum-rate using the PC-TIN scheme, and the sum-capacity upper bounds $\overline{C}_{\Sigma 1}$ in \eqref{UB1} and $\overline{C}_{\Sigma 2}$ in \eqref{UB2}.

\begin{figure}[t]
\centering
\psfragscanon
\psfrag{x}[l]{$h$}
\psfrag{y}[b]{Sum Rate (bits/channel use)}
\psfrag{L1}[l]{\footnotesize{$\overline{C}_{\Sigma 1}$}}
\psfrag{L2}[l]{\footnotesize{$\overline{C}_{\Sigma 2}$}}
\psfrag{L3}[l]{\footnotesize{SD-TIN}}
\psfrag{L4}[l]{\footnotesize{TDMA-TIN}}
\psfrag{L5}[l]{\footnotesize{PC-TIN}}
\includegraphics[width=0.9\columnwidth]{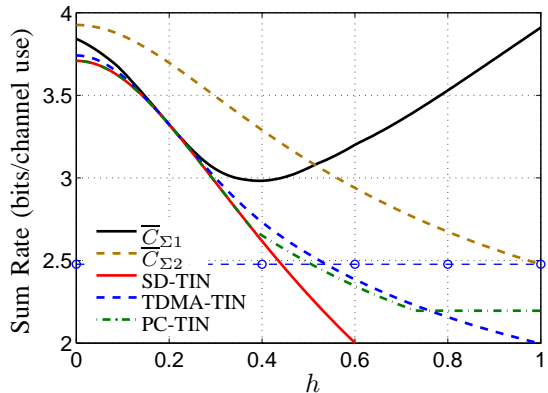}
\caption{Achievable sum-rate and upper bound for a PIMAC with $h_{22}=0.2$, $h_{12}=h_{31}=h$, and $P=10$.}
\label{SumRate1}
\end{figure}

Notice that the TDMA-TIN scheme always outperforms the SD-TIN scheme and that the achievable sum-rate is close to the upper bound at low values of $h$ (weak interference). The achievable rate using PC-TIN is optimized using exhaustive search. It can be seen that the PC-TIN schemes follows three different trends in this figure. Namely, for low values of $h$ (in this case $h\in[0,0.39]$), the sum-rate $C(p_1,p_2,p_3)$ in \eqref{C_PC-TIN} is maximized at $(p_1,p_2,p_3)=(P_1,P_2,P_3)$. That is, at low interference all users are active, and transmit with full power, and thus PC-TIN performs the same as the SD-TIN scheme. As $h$ increases, the optimal $(p_1,p_2,p_3)$ becomes $(0,P_2,P_3)$ and PC-TIN outperforms SD-TIN. This means that in this case, the first transmitter is silent, and the other two transmitters use full power. Notice that in this case, the IC formed by users 2 and 3 is 'less noisy' than that formed by users 1 and 3, and hence achieves higher sum-rate. As $h$ grows further, the optimal $(p_1,p_2,p_3)$ becomes $(P_1,P_2,0)$, and PC-TIN outperforms TDMA-TIN. In this case, only the MAC users are active, and the PIMAC is operated as a MAC since the third transmitter is kept silent. Notice however that in the last case there is no interference, and thus no TIN. The bottom line (with circle markers) is the achievable sum-rate using a scheme with TDMA between the MAC and the P2P users achieving
\begin{align}
R_\Sigma^{TDMA}&=\max_{\alpha\in[0,1]}\frac{\alpha}{2}\log\left(1+\frac{P_1}{\alpha}+\frac{P_2}{\alpha}\right)\nonumber\\
&\hspace{1.2cm}+\frac{1-\alpha}{2}\log\left(1+\frac{P_3}{1-\alpha}\right),
\end{align}
which is maximized with
\begin{align}
\alpha=\frac{P_1+P_2}{P_1+P_2+P_3}.
\end{align}

The upper bound $\overline{C}_{\Sigma 1}$ and the lower bounds coincide at $h=h_{22}=0.2$. In this case, the PIMAC can be modeled as an IC with one transmitter sending $X_1^n+X_2^n$. In this case, $\overline{C}_{\Sigma 1}$ is indeed the sum-capacity of the PIMAC as indicated in \cite[Remark 1]{ChaabanSezgin_EW2011}. The new upper bound $\overline{C}_{\Sigma 2}$ becomes tighter as $h$ increases and coincides with the TDMA lower bound at $h=1$ as shown in the Figure.

\section{Conclusion}
\label{Conclusion}
We have examined treating interference as noise in a network consisting of interfering point-to-point channel and a multiple access channel. We showed that the achievable scheme where the users transmit at full power and the receivers treat interference as noise is suboptimal. It is outperformed by a scheme which combines TDMA with treating interference as noise. We also derived a new sum-capacity upper bound. Although the known upper bounds do not coincide with the lower bounds, they are nearly tight if interference is weak. Tightening these bounds would be a topic for future work, in order to characterize the sum-capacity of the network at noisy/weak interference.

\bibliography{myBib}

\end{document}